\documentclass[a4paper,12pt]{article}
\usepackage{amssymb,amsfonts,amsmath,mathtext,cite,enumerate,float}
\usepackage{amssymb,amsthm,latexsym,euscript,amscd, authblk}
\usepackage{braket}

\usepackage[english]{babel}
\usepackage[utf8]{inputenc}

\newtheorem{proposition}{Proposition}

\newtheorem{theorem}{Theorem}

\title{On perturbations of dynamical semigroups defined by covariant completely positive measures on the semi-axis}

\author[1]{G.G. Amosov\thanks {gramos@mi-ras.ru}}

\affil[1] {Steklov Mathematical Institute of Russian Academy of Sciences, Gubkina str., 8, Moscow 119991, Russia}

\begin{document}
	
	\maketitle
	
	\begin{abstract}
     We consider perturbations of dynamical semigroups on the algebra of all bounded operators in a Hilbert space generated by covariant completely positive measures on the semi-axis. The construction is based upon unbounded linear perturbations of generators of the preadjoint semigroups on the space of nuclear operators. As an application we construct a perturbation of the semigroup of non-unital *-endomorphisms on the algebra of canonical anticommutation relations resulting in the flow of shifts.
	\end{abstract}
	
	Keywords: perturbations of dynamical semigroups, covariant completely positive measures on the semi-axis, the flow of shifts on the algebra of canonical anticommutation relations

\section{Introduction}

The theory of one-parameter $C_0$-semigroups (strong continuous) of linear transformations $T_t:X\to X,\ t\ge 0,$ on the Banach space $X$ introduced in the pioneering papers \cite {Hille1948, Yosida1948} states the conditions for the closed linear operator $\mathcal L$ with a dense domains $D({\mathcal L})\subset X$ to be a generator of $T=\{T_t,\ t\ge 0\}$ such that $T_t=exp(t\mathcal {L}),\ t\ge 0$. If $X=\mathfrak {S}_1(H)$ is the Banach space of nuclear operators in a Hilbert space $H$  the claim of strong continuity for orbits of $T$ possessing the property of non-increasing a trace is equivalent to weak continuity that is $Tr(T_t(\rho )x)\to Tr(\rho x)$ if $t\to 0$ for all $\rho \in \mathfrak {S}_1(H)$ and $x\in B(H)$ (the algebra of all bounded operators in $H$) \cite {DAntonio1967}. In \cite{Holevo1995} it was shown that the perturbation of the generator $\mathcal L$ by a linear map $\Delta $ satisfying some additional conditions can be represented in the form of integral equation including the operator-valued measure generated by $\Delta $.  Together with a perturbation of $T$ it is naturally to consider the corresponding perturbation of the adjoint semigroup $T^*=\{T_t^*,\ t\ge 0\}$ on the algebra $B(H)$. We realize this construction starting directly with the measure.  As an example we construct a perturbation of
the semigroup of non-unital *-endomorphisms on the algebra of canonical anticommutation relations (CAR). As a result of this perturbation we obtain the flow of shifts on the CAR algebra \cite {Powers1988}. Earlier we have announce our result for the CAR algebra in \cite {Kholmogorov}.  Note that the perturbations of a semigroup on $B(H)$ generated by the perturbation of the generator $\mathcal L$ of the corresponding preadjoint semigroup on $\mathfrak {S}_1(H)$ by a linear map with the domain containing $D({\mathcal L})$ form the basis for the construction of non-standard quantum dynamical semigroups \cite {Holevo2017, Holevo2018}.

\section {Preliminaries}

Let $\Psi :\mathfrak {S}_1(H)\to \mathfrak {S}_1(H)$ be a linear bounded map on the Banach space of nuclear operators $\mathfrak {S}_1(H)$ in a Hilbert space $H$. Since $(\mathfrak {S}_1(H))^*=B(H)$ (the algebra of all bounded operators in $H$) the adjoint map $\Phi=\Psi ^*:B(H)\to B(H)$ is weak* continuous. $\Psi $ is said to be preadjoint of $\Phi $ and denoted $\Psi =\Phi _*$. A one-parameter family of linear *-maps $\Phi _t:B(H)\to B(H),\ t\ge 0,$ is said to be a dynamical semigroup if

(i) $\Phi _{t+s}=\Phi _t\circ \Phi _s,\ t,s\ge 0,\ \Phi _0=Id$;

(ii) each $\Phi _t$ is completely positive and $\Phi _t(I)\le I,\ t\ge 0$;

(iii) $Tr(\rho \Phi _t(x))$ is continuous in $t$ for all $x\in B(H),\ \rho \in \mathfrak {S}_1(H)$.

\noindent
In the case, $\{\Psi _t=(\Phi _t)_*,\ t\ge 0\}$ is a $C_0$-semigroup on $\mathfrak {S}_1(H)$ with the property $Tr(\Psi _t(\rho ))\le Tr(\rho )$ for all $\rho >0,\ \rho \in \mathfrak {S}_1(H)$. Thus, there is a generator $\mathcal L$ with the dense domain $D(\mathcal L)\subset \mathfrak {S}_1(H)$ such that $\rho _t=\Phi _t(\rho ),\ \rho \in \mathfrak {S}_1(H)$ is a solution to the Cauchy problem
$$
\frac {d\Psi _t(\rho _t)}{dt}={\mathcal L}(\rho _t),\ t>0,
$$
$$
\rho _0=\rho .
$$
Following to \cite {Holevo1995} let us define a perturbation of $\mathcal L$ of the form
$$
\check {\mathcal L}=\mathcal {L}+\Delta ,
$$
where the linear map $\Delta: dom{\mathcal L}\to \mathfrak {S}_1(H)$ satisfies the properties

(i) for any positive definite matrix $||\rho _{jk}||,\ \rho _{jk}\in dom{\mathcal L}$ the matrix $||\Delta (\rho _{jk})||$ is
positive definite;

(ii) 
$
Tr(\Delta (\rho ))\le -Tr({\mathcal L}(\rho )),\ \rho \in dom{\mathcal L},\ \rho > 0.
$

Consider the measure ${\mathcal M}_*$ with values in the set of complete positive maps on $\mathfrak {S}_1(H)$ determined by the formula
\begin{equation}\label{measure}
{\mathcal M}_*([t,s))=\int \limits _t^s\Delta \circ \Psi _rdr.
\end{equation} 
By a construction the measure (\ref {measure}) satisfies the relation
\begin{equation}\label{trans1}
{\mathcal M}_*([t,s))\circ \Psi _r={\mathcal M}_*([t+r,s+r)),\ s,t,r\ge 0.
\end{equation}

Then \cite {Holevo1995} the equation 
\begin{equation}\label{integral}
\frac {d}{dt}Tr(\rho \check \Phi _t(x))=Tr(({\mathcal L}+\Delta )(\rho )\check \Phi_t(x)),\ \rho \in dom{\mathcal L},\ x\in B(H),
\end{equation}
is equivalent to the integral equation
\begin{equation}\label{integral2}
\check \Phi _t-\int \limits _{0}^t{\mathcal M}(dt)\circ \check\Phi _{t-s}=\Phi _t,\ t\ge 0,
\end{equation}
where the measure ${\mathcal M}$ consists of maps on $B(H)$ adjoint to (\ref {measure}). 
Due to (\ref {trans1}) the measure ${\mathcal M}$ has the covariant property
\begin{equation}\label{trans2}
\Phi _r\circ {\mathcal M}([t,s))={\mathcal M}([t+r,s+r)),\ t,s,r\ge 0.
\end{equation}
Given two completely positive maps $\Theta _1$ and $\Theta _2$ on $B(H)$ we shall use the notation
$$
\Theta _1\succ \Theta _2
$$
iff
$$
\Theta _1-\Theta _2
$$
is completely positive.
In \cite{Holevo1995} it was considered the measures satisfying (\ref {trans2}) generated by the formula
\begin{equation}\label{1995}
{\mathcal M}([t,s))=\Phi _t\circ \Theta -\Phi _s\circ \Theta ,
\end{equation}
where $\Theta $ is the excessive completely positive map in the sense that
$$
\Theta \succ \Phi _t\circ \Theta ,\ t>0.
$$ 
In the case, it was shown that (\ref {integral}) has a unique minimal solution $\check \Phi^{\infty }_t$ possessing the property that any other solution $\check \Phi_t$ satisfies
$$
\check \Phi _t\succ \check \Phi ^{\infty }_t,\ t\ge 0.
$$

\section {Perturbations generated by measures}

We consider measures ${\mathcal M}$ on Borel subsets of the semi-axis ${\mathbb R}_+$ such that given $0\le t\le s\le +\infty $

(i) {${\mathcal M}([t,s)):\ {B}(H)\to B(H)$ is a completely positive linear map;}

(ii) ${\mathcal M}([t,s))(I)\le I$;

(iii) $\mu _{\rho ,x}([t,s))=Tr(\rho \mathcal {M} ([t,s))(x))$ is a $\sigma $-additive measure on $\mathbb {R}_+$ for any fixed $\rho \in \mathfrak {S}_1(H)$ and $x\in B(H)$.

Together with $\mathcal {M} ([t,s))$ it is naturally to examine a preadjoint map $\mathcal {M} _*([t,s)):\mathfrak {S}_1(H)\to \mathfrak {S}_1(H)$ possessing the property $Tr({\mathcal M}_*([t,s))(\rho ))\le Tr(\rho )$ for all positive $\rho \in \mathfrak {S}_1(H)$.

Notice that if additionally $\mathcal {M} ({\mathbb R})(I)=I$, then $\mathcal M$ is said to be a completely positive instrument \cite {Davies}. Nevertheless this property should not take place for our purposes.

The measure $\mathcal M$ is said to be covariant with respect to dynamical semigroup $\Phi $ if (\ref {trans2}) holds true. 
It is not clear whether each covariant measure $\mathcal M$ satisfying (\ref {trans2}) can be obtained from some excessive map $\Theta $ by means of (\ref {measure}). Developing the techniques of \cite {Holevo1995} we suggest replacing (\ref {1995}) to
an arbitrary measure $\mathcal M$ with values in the set of completely positive maps satisfying (\ref {trans2}).

\begin{proposition}\label{1}
Given a meaure $\mathcal M$ covariant with respect to dynamical semigroup $\Phi $ there exists a minimal solution to (\ref {integral2}). 
\end{proposition}

\begin{proof}
Following to \cite {Holevo1995} let us consider the iteration process
$$
\Phi _t^{n+1}=\Phi _t^{n}+\int \limits _0^t\mathcal {M}(dr)\circ \Phi _r^n.
$$
Due to $\Phi _t(I)\le I$ and $\mathcal {M}([t,s))(I)\le I,\ 0\le t\le s\le +\infty$ we get
$$
\Phi _t^{n+1}\succ \Phi _t^n,\ \Phi _t^n(I)\le I.
$$
It results in $\Phi _t^n$ tends to the minimal solution $\check \Phi _t$ of (\ref {integral2}).
\end{proof}

\section{Perturbations of no-event semigroups}

Following to \cite {Holevo2017} $\Psi^0 =\{\Psi ^0_t:\mathfrak {S}_1(H)\to \mathfrak {S}_1(H),\ t\ge 0\}$ is said to be a no-event semigroup if every pure state $\rho =\ket {\psi }\bra {\psi }$ is mapped to a multiple of a pure state. Such a semigroup necessary has the form
$$
\Psi _t(\rho )=T_t\rho T_t^*,\ t\ge 0,
$$
where $T_t=exp(tK)$ is a $C_0$-semigroups of contractions. Hence the generator $K$ is a maximum dissipative operator due to \cite {LF1961}. The generator $\mathcal {L}$ of $\Psi $ is acting by the formula
$$
\mathcal {L}(\rho )=K\rho +\rho K,\ \rho \in D({\mathcal L}),
$$
where the domain $D({\mathcal L})$ includes rank one operators $\ket {\psi }\bra {\xi},\ \psi ,\xi \in D(K)$. 
Take linear operators $L_j:D(K)\to H$ possessing the property
$$
\sum \limits _j||L_j\psi ||^2\le -2Re\braket {\psi ,K\psi }
$$
and define a linear map $\Delta :D({\mathcal L})\to \mathfrak {S}_1(H)$  by the formula
$$
\Delta (\ket {\psi }\bra {\xi})=\sum \limits _j\ket {L_j\psi }\bra {L_j\xi },
$$
$\psi ,\xi \in D(K)$.
Consider the measure determined by (\ref {measure}). Then, it has the form (\ref {1995}), where the excessive completely positive map $\Theta$ is given by the formula (\cite {Holevo1995}, Lemma 2)
$$
\braket {\psi ,\Theta (x)\psi }=\int \limits _0^{+\infty }\sum \limits _j\braket {L_jT_t\psi ,xL_jT_t\psi }dt,\ \psi \in D(K),\ x\in B(H).
$$
Below we shall give an example of perturbation for a no-event semigroup on the algebra of canonical anticommutation relations.

\section{Algebra of canonical anticommutation relations $\mathfrak {A}(H)$}

Here we record the basic concepts about the algebra of canonical anticommutation relations \cite {Brat}.
Let $H$ be a separable infinite dimensional Hilbert space. Fix the orthonormal basis $(\ket {j})_{j=1}^{+\infty }$ in $H$. Then, the antisymmetric Fock space $F(H)$ over one-particle
Hilbert space $H$ is a Hilbert space with the orthonormal basis $\ket {0},\ \ket {j_1\dots j_n}$, where the indices $j_1<j_2<\dots <j_n$ and $n$ run over the set $\{1,2,3,\dots \}$. The vector
$\ket {0}$ is said to be vacuum. Let us define the ladder operators $a_k^{\dag},\ a_k$ by the formula
$$
a_k^{\dag }\ket {j_1\dots j_n}=\left \{\begin{array}{c}(-1)^s\ket {j_1\dots j_{s}kj_{s+1}\dots j_n}\ if\ j_s<k<j_{s+1}\\ 0\hskip 2cm if\ k\in \{j_1,\dots j_n\}\end{array}\right .,
$$ 
$$
a_k\ket {j_1\dots j_n}=\left \{\begin{array}{c}(-1)^{s+1}\ket {j_1\dots j_{s-1}j_{s+1}\dots j_n}\ if\ j_s=k\\0\hskip 2cm if\ k\notin \{j_1,\dots j_n\}.\end{array}\right .,
$$
$$
a_k^{\dag }\ket {0}=\ket {k},\ a_k\ket {0}=0,\ k=1,2,3,\dots 
$$
It follows that
$$
a_ka_k^{\dag}+a_k^{\dag}a_k=I,\ (a_k)^2=(a_k^{\dag })=0,
$$
$$
a_ka_j=-a_ja_k,\ a_k^{\dag }a_j^{\dag }=-a_j^{\dag }a_k^{\dag }.
$$
The $C^*$-algebra $\mathfrak {A}(H)$ generated by the ladder operators is said to be the algebra of canonical anticommutation relations (CAR) in the Fock representation. The CAR algebra $\mathfrak {A}(H)$ is generated by monomials $x_{j_1}\dots x_{j_n}$, where
$j_1<j_2<\dots <j_n$ and $x_{j_s}\in \{a^{\dag }_{j_s},a_{j_s},a_{j_s}a^{\dag }_{j_s}\}$.

Let us define a linear map on rank one operators by the formula
$$
\Xi _*(\ket {j_1\dots j_n}\bra {r_1\dots r_m})=
$$
\begin{equation}\label{ups5}
\sum \limits _{s,k}(-1)^{s+k}\delta _{j_sr_k}\ket {j_1\dots j_{s-1}j_{s+1}\dots j_n}
\bra {r_1\dots r_{k-1}r_{k+1}\dots r_m},
\end{equation}
$$
\Xi _*(\ket {0}\bra {r_1\dots r_m})=\Xi _*(\ket {j_1\dots j_n}\bra {0})=\Xi _*(\ket {0}\bra {0})=0.
$$
Notice that (\ref {ups5}) is the sum of partial traces over minimal subsystems \cite{Filip}.

\begin{proposition}\label{ups3}
Formula (\ref {ups5}) correctly determines a linear map on $\mathfrak {S}_1(F(H))$ which can be uniquely extended  to the completely positive map on $B(F(H))$. This map doesn't have the property of non-increasing a trace.
\end{proposition}

\begin{proof}

It is straightforward to check that
$$
\Xi _*(\ket {j_1\dots j_n}\bra {r_1\dots r_m})=\sum \limits _ka_k\ket {j_1\dots j_n}\bra {r_1\dots r_m}a_k^{\dag }.
$$
It follows that $\Xi _*$ can be uniquely extended to a completely positive map on $B(F(H))$.
Denote
$$
Q=\sum \limits _ka_k^{\dag }a_k.
$$
Since
$$
a_k^{\dag }a_k\ket {j_1\dots j_n}=\sum \limits _{s}\delta _{kj_s}\ket {j_1\dots j_n},\ a_k^{\dag }a_k\ket {0}=0
$$
for any $j_1<j_2<\dots <j_n$
we get
$$
Q\ket {j_1\dots j_n}=n\ket {j_1\dots j_n}.
$$
Hence $Tr(\Xi _*(\ket {j_1\dots j_n}\bra {j_1\dots j_n}))=Tr(Q\ket {j_1\dots j_n}\bra {j_1\dots j_n})=n$ and $\Xi _*$ has not the property of non-increasing a trace.

\end{proof}

Given $f=\sum\limits _jc_j\ket {j},\ c_j\in {\mathbb C},\ \sum \limits _j|c_j|^2<+\infty $, let us define the ladder operators $a^{\dag }(f),a(f)$ over $f\in H$ by the formula
$$
a(f)=\sum \limits _j\overline {c_j}a_j,\ a^{\dag }(f)=\sum \limits _{j}c_ja^{\dag }_j
$$
satisfying the relations
$$
a(f)a^{\dag}(g)+a^{\dag}(g)a(f)=\braket {f,g}I,
$$
$$
a(f)a(g)+a(g)a(f)=a^{\dag }(f)a^{\dag }(g)+a^{\dag }(g)a^{\dag }(f)=0.
$$
It follows from the definition that the ladder operators
$$
||a(f)||=||a^{\dag }(f)||=||f||
$$
in the $C^*$-algebra $\mathfrak {A}(H)$.

Let us introduce the outer multiplication $\Lambda $ over indexes $j_1\dots j_n$ such that
if $j_1<j_2<\dots <j_n$, then
$$
j_1\Lambda j_2\Lambda \dots \Lambda j_n=\ket {j_1\dots j_n}
$$
and
$$
j_s\Lambda j_k=-j_k\Lambda j_s.
$$
Following this way, for $f_j=\sum \limits _kc_{jk}\ket {k},\ \sum \limits _k|c_{jk}|^2<+\infty,$ we can put
$$
f_1\Lambda \dots \Lambda f_n=\sum \limits _{k_1,\dots ,k_n}c_{1k_1}\dots c_{nk_n}j_{k_1}\Lambda \dots \Lambda j_{k_n}
$$
Hence, given $f_j,g_k\in H$ we can define vectors $\ket {\bf f}=f_{1}\Lambda \dots \Lambda f_{n},\ket {\bf g}=g_1\Lambda \dots \Lambda g_n\in H^{\otimes _a^{n}}$. Fix $n$ and denote $H^{\otimes _a^n}$ the closed linear envelope of vectors $\ket {\bf f}$ in the Fock space $F(H)$.
Then, a restriction of the inner product in $F(H)$ to $H^{\otimes _a^n}$ reads 
$$
\braket {{\bf f}|\bf {g}}_{H^{\otimes _a^n}}=det||\braket {f_j|g_k}||,
$$
where the outer multiplication $\Lambda $ satisfies the rule
$$
f\Lambda g=-g\Lambda f,\ f,g\in H.
$$

Alternatively we can define the orthogonal projection $P_a$ in a tensor product $H^{\otimes ^n}$ as follows
$$
P_a(f_1\otimes \dots \otimes f_n)=\frac {1}{n!}\sum \limits _{\epsilon \in S_n}(-1)^{|\epsilon|}f_{\epsilon (1)}\otimes \dots \otimes f_{\epsilon(n)},
$$
where the sum is taken over the set of all permutations $S_n$ and $|\epsilon |$ is a signature of permutation $\epsilon \in S_n$. By this way, 
$$
H^{\otimes _a^n}=P_a(H^{\otimes n})
$$
and is said to be an n-th antisymmetric tensor product of $H$.

Given ${\bf f}\in H^{\otimes _a^n}$ we denote
$$
a({\bf f})=a(f_1)\dots a(f_n),\ a^{\dag }({\bf f})=a^{\dag }(f_1)\dots a^{\dag }(f_n).
$$
It follows that
$$
a^{\dag }({\bf f}){\bf g}={\bf f}\Lambda {\bf g}
$$
and
$$
a(f){\bf g}=\sum \limits _k(-1)^{k+1}\braket {f,g_k}g_1\Lambda \dots \Lambda g_{k-1}\Lambda g_{k+1}\Lambda \dots \Lambda g_n,
$$
where $f\in H,\ {\bf g}\in H^{\otimes _a^n}$.

\section{The semigroup of shifts on $\mathfrak {A}(H)$}

Put $H=L^2({\mathbb R}_+)$ and  define the semigroup of shifts in $H$ by the formula
$$
(S_tf)(x)=\left \{\begin{array}{c}f(x-t),\ x>t;\\ 0,\ 0\le x\le t,\end{array}\right .
$$
$t\ge 0,\ f\in H$.
The conjugate semigroup of contractions $S_t^*=e^{td}$ has a generator
$$
df=\frac {df}{dx},\ f\in D(d)=\{f\ |\ f'\in L^2({\mathbb R}_+)\}.
$$
We also need the semigroup of shifts in $F(H)$ obtained by lifting $(S_t)$ as follows
\begin{equation}\label{shift}
\hat {S}_t(f_1\Lambda \dots \Lambda f_n)=S_tf_1\Lambda \dots \Lambda S_tf_n,\ \hat {S}_t\ket{0} =\ket{0} ,
\end{equation}
$t\ge 0,\ f_j\in H$.
The generator $\hat d$ of its conjugate semigroup of contractions $\hat {S}_t^*=e^{it\hat d}$ is given by the formula
\begin{equation}\label{d}
\hat d\ket {\bf f}=\sum \limits _{j=1}^nf_1\Lambda \dots \Lambda f_{j-1}\Lambda df_j\Lambda f_{j+1}\Lambda \dots \Lambda f_n,\ \hat d\ket{0} =0,
\end{equation}
$f_j\in D(d)$. Along (\ref {d}) we need a preconjugate operator acting by the formula
$$
\hat {d}_*\ket {\bf f}=-\sum \limits _{j=1}^nf_1\Lambda \dots \Lambda f_{j-1}\Lambda df_j\Lambda f_{j+1}\Lambda \dots \Lambda f_n,\ \hat {d}_*\ket{0} =0,
$$
$f_j\in D(d_*)=\{f\ |\ f'\in L^2({\mathbb R}_+),\ f(0)=0\}$.
Using (\ref {shift}) it is possible to determine the dynamical semigroup on $B(F(H))$ as follows
\begin{equation}\label{S0}
\Phi _t(x)=\hat S_tx\hat S_t^*,\ 
\end{equation}
$t\ge 0,\ x\in B(F(H))$.
The preadjoint semigroup $\Psi _t:\mathfrak {S}_1(F(H))\to \mathfrak {S}_1(F(H))$ is given by the formula
\begin{equation}\label{S1}
\Psi _t(\rho )=\hat S_t^*\rho \hat S_t,
\end{equation}
$t\ge 0,\ \rho \in \mathfrak {S}_1(F(H))$. Note that (\ref {S1}) can be directly extended to the semigroup of non-unital *-endomorphisms on $B(F(H))$.
The generator $\mathcal L$ of (\ref {S1}) is determined by the formula
$$
\mathcal {L}(\rho )=[\hat d,\rho ]=\hat d\rho -\rho \hat d,
$$
where $\rho$ belongs to the domain $D({\mathcal L})$ which is dense in $\mathfrak {S}_1(F(H))$. 
It is straightforward to see that $D(\mathcal {L})$ contains rank one operators
$$
\ket {{\bf f}}\bra {{\bf g}},\ \ket {{\bf f}}\bra {0},\ \ket {0}\bra {{\bf g}},\ \ket {0}\bra {0},
$$
where $f_j,g_k\in D(d)$.

The semigroup of unital $*$-endomorphisms $\check \Phi _t$ on $\mathfrak {A}(H)$ defined by the relation
\begin{equation}\label{powers}
\check {\Phi} _t(a(f))=a(S_tf),\ t\ge 0,
\end{equation}
is said to be the flow of shifts on the CAR algebra \cite {Powers1988}.
Denote ${\bf f}_{\backslash j}=f_1\Lambda \dots \Lambda f_{j-1}\Lambda f_{j+1}\Lambda \dots \Lambda f_n$ and define
a linear $*$-map on $D(\mathcal L)$ by the formula
\begin{equation}\label{Delta}
\Delta (\ket {\bf f}\bra {\bf g})=\sum \limits _{j,k}(-1)^{j+k}f_j(0)\overline {g}_k(0)\ket {{\bf f}_{\backslash j}}\bra {{\bf g}_{\backslash k}}, 
\end {equation}
$f_j,g_k\in D(d)$.

\begin{theorem}\label {1}
$$
Tr((\mathcal {L}(\ket {\bf f}\bra {\bf g})+\Delta (\ket {\bf f}\bra {\bf g}))a({\bf h})a^{\dag }({\bf e}))=
\bra {\bf f}(a(\hat d_*{\bf h})a^{\dag }({\bf e})+a({\bf h})a^{\dag }(\hat d_*{\bf e}))\ket {\bf g}.
$$
$f_j,g_k\in D(d),\ h_j,e_k \in D(d_*)$.
\end{theorem}
\begin{proof}
The trace
$$
Tr((\mathcal {L}(\ket {\bf f}\bra {\bf g})a({\bf h})a^{\dag }({\bf e}))
$$
can be represented as a sum of elements given as a multiplication of inner products of $h_l,e_m$ and $f_j,g_k$ as well as the derivatives of $f_j,g_k$ such that only at least one of them could contain a derivative of the following possible forms
\begin{equation}\label{sl1}
\braket {f_j'|e_k},\ \braket {f_j'|h_k},\ \braket {e_k|g_j'},\ \braket {h_k|g_j'}
\end{equation}
or
\begin{equation}\label{sl2}
\braket {f_j'|g_k},\ \braket {f_j|g_k'}.
\end{equation}
If (\ref {sl1}) is implemented, then the derivative can be passed to the other side because $e _j,h _k\in D(d_*)$
resulting in $e_j(0)=h_k(0)=0$, e.g. $\braket {f_j'|e_k}=-\braket {f_j|e_k'}$. Taking integration by parts in
(\ref {sl2}) we obtain the term outside the integral of the form $f_j(0)\overline {g_k}(0)$ but it is self destructing with
the corresponding term in 
$$
Tr(\Delta (\ket {\bf f}\bra {\bf g})a({\bf h})a^{\dag }({\bf e})).
$$

More formally, 
\begin{equation}\label{anar1}
Tr(\mathcal {L}(\ket {\bf f}\bra {\bf g})a({\bf h})a^{\dag }({\bf e}))=\braket {a^{\dag }({\bf h}){\bf g},a^{\dag }({\bf e})\hat d{\bf f}}+\braket {a^{\dag }({\bf h})\hat d{\bf g},a^{\dag }({\bf e}){\bf f}}
\end{equation}
The first term in (\ref {anar1}) can be rewritten as
\begin{equation*}
\braket {a^{\dag }({\bf h}){\bf g},a^{\dag }({\bf e})\hat d{\bf f}}=\sum \limits _{j=1}^n\braket {{\bf h}\Lambda {\bf g},
{\bf e}\Lambda f_1\Lambda \dots f_{j-1}\Lambda df_j\Lambda f_{j+1}\Lambda \dots \Lambda f_n}=
\end{equation*}
\begin{equation*}
\braket {{\bf h}\Lambda {\bf g},\hat d({\bf e}\Lambda {\bf f})}-\braket {
{\bf h}\Lambda {\bf g},(\hat d{\bf e})\Lambda {\bf f}}=
\end{equation*}
\begin{equation}\label{anar2}
\braket {{\bf h}\Lambda {\bf g},\hat d({\bf e}\Lambda {\bf f})}+\braket {a^{\dag }({\bf h}){\bf g},a^{\dag }(\hat d_*{\bf e}){\bf f}}
\end{equation}
because $\hat d_*{\bf e}=-\hat d{\bf e}$ if ${\bf e}\in D(\hat d_*)$. Integrating by parts the first term in (\ref {anar2}) we obtain
\begin{equation*}
\braket {{\bf h}\Lambda {\bf g},\hat d({\bf e}\Lambda {\bf f})}=\braket {(\hat d_*{\bf h})\Lambda {\bf g},{\bf e}\Lambda {\bf f}}-
\braket {{\bf h}\Lambda (\hat d{\bf g}),{\bf e}\Lambda {\bf f}}-
\end{equation*}
\begin{equation}\label{anar3}
\sum \limits _{j,k}(-1)^{j+k}f_j(0)\overline {g}_k(0)\braket {{\bf h}\Lambda {\bf g}_{\backslash k},{\bf {e}\Lambda {\bf f}}_{\backslash j}}
\end{equation}
due to $e_j(0)=h_k(0)=0$ in virtue of ${\bf e},{\bf h}\in D(\hat d_*)$. Substituting (\ref {anar3}) to (\ref {anar2}) we get
\begin{equation*}
\braket {a^{\dag }({\bf h}){\bf g},a^{\dag }({\bf e})\hat d{\bf f}}=\braket {a^{\dag }(\hat d_*{\bf h}){\bf g},a^{\dag }({\bf e}){\bf f}}-\braket {a^{\dag }({\bf h})\hat d{\bf g},a^{\dag }({\bf e}){\bf f}}-
\end{equation*} 
\begin{equation}\label{anar4}
Tr(\Delta (\ket {{\bf f}}\bra {{\bf g}})a({\bf h})a^{\dag }({\bf e}))+\braket {a^{\dag }({\bf h}){\bf g},a^{\dag }(\hat d_*{\bf e}){\bf f}}
\end{equation}
Comparing (\ref {anar1}) and (\ref {anar4}) completes the proof.

\end{proof}

Let us define the map ${\mathcal M}_{*}$ on measurable sets on ${\mathbb R}_+$ with values in the set of linear maps defined on rank one operators $\ket {\bf f}\bra {\bf g}$ by the formula
\begin{equation}\label{M}
{\mathcal M }_*([t,s))(\ket {\bf f}\bra {\bf g})=\sum \limits _{j,k}(-1)^{j+k}\int \limits _t^sdr
f_j(r)\overline {g_k}(r)\Psi _r(\ket {{\bf f}_{\backslash j}}\bra {{\bf g}_{\backslash k}}),
\end{equation}
$$
{\mathcal M} _*([t,s))(\ket {0}\bra {\bf f})={\mathcal M} _*([t,s))(\ket {\bf f}\bra {0})={\mathcal M} _*([t,s))(\ket {0}\bra {0})=0.
$$
Notice that formally
$$
{\mathcal M}_*([t,s))=\int \limits _t^sdr\Delta \circ \Psi _r,
$$
where $\Delta $ is defined by (\ref {Delta}). Moreover,
\begin{equation}\label{mera}
\lim \limits _{t\to 0}{\mathcal M}_*([0,t))(\ket {\bf f}\bra {\bf g})=\Delta (\ket {\bf f}\bra {\bf g})
\end{equation}
for any choice of $f_j,g_k\in D(d)$.

\begin{proposition}\label{cpi}
The map $\mathcal M$ conjugate to (\ref {M}) is the measure covariant with respect to $\Phi $.
\end{proposition}

\begin{proof}
 
Analogously to the proof of Proposition \ref {ups3} let us define a completely positive map by the formula
\begin{equation}\label{r1}
\Xi _*(\delta )(\rho )=a(\chi _{[0,\delta ]})\rho a^{\dag }(\chi _{[0,\delta ]}),\ \rho \in \mathfrak {S}_1(H),
\end{equation}
where $\chi _{[0,\delta ]}$ is a characteristic function of the segment $[0,\delta ]$ and $\delta >0$. 
Given $0\le t<s$ and integer $n$ consider the auxiliary completely positive map
\begin{equation}\label{r2} 
{\Sigma}_{n}=\sum \limits _{j=1}^n\Xi _*\left (\frac {s-t}{n}\right )\circ \Psi _{\frac {(s-t)j}{n}}.
\end{equation}
For the Kraus operators $V_{j}=a(\chi _{[0,\frac {s-t}{n}]})\hat S^*_{\frac {(s-t)j}{n}}$ of ${\Sigma}_{n}$ let us examine the sum
$$
Q_n=\sum \limits _{j=1}^nV_{j}^*V_{j}=\sum \limits _{j=1}^n\hat S_{\frac {(s-t)j}{n}}a^{\dag }(\chi _{[0,\frac {s-t}{n}]})a(\chi _{[0,\frac {s-t}{n}]})\hat S^*_{\frac {(s-t)j}{n}}.
$$
Taking into account
$$
||a^{\dag }(\chi _{[0,\frac {s-t}{n}]})a(\chi _{[0,\frac {s-t}{n}]})||=\frac {s-t}{n}
$$
we obtain
$$
Q_n<(s-t)I.
$$
It follows that ${\Sigma }_n$ is non-increasing a trace.
Hence
\begin{equation}\label{r3}
{\mathcal M}_*([t,s))=\lim \limits _{n\to +\infty }{\Sigma }_n
\end{equation} 
is a completely positive map and $Tr({\mathcal M}_*([t,s))(\rho ))\le Tr(\rho)$ for all positive $\rho \in \mathfrak {S}_1(H)$.

\end{proof}

The flow of shifts $\check {\Phi }_t$ determined by (\ref {powers}) has the generator $\check {\mathcal L}$ acting by the formula
\begin{equation}\label{gen}
\check {\mathcal L}(a({\bf h})a^{\dag }({\bf e}))=a(\hat d_*{\bf h})a^{\dag }({\bf e})+a({\bf h})a^{\dag }(\hat d_*{\bf e})
\end{equation}
for any $h_j,e_k\in D(d_*)$. On the other hand, it follows from Theorem 1 that the generator $\check {\mathcal L}_*$ of the preadjoint semigoup $\check \Psi _t=(\check \Phi _t)_*$ is equal to ${\mathcal L}+\Delta $.
In the next theorem we show that $\check {\Phi }_t$ satisfies the integral equation.

\begin{theorem}\label {2}
The flow of shifts (\ref {powers}) is a solution to the integral equation
\begin{equation}\label{main}
\check {\Phi }_t-\int \limits _0^t{\mathcal M}(ds)\circ \check {\Phi }_{t-s}=\Phi _t,\ t\ge 0,
\end{equation}
where ${\mathcal M}([t,s)):B(F(H))\to B(F(H))$ is the measure determined in Proposition {\ref {cpi}}.
\end{theorem}

\begin{proof}

The solution to (\ref {main}) exists and defines a dynamical semigroup due to Proposition \ref {1}. 
Apply the left hand side of (\ref {main}) to $a({\bf h})a^{\dag }({\bf e})$, multiply to $\ket {\bf f}\bra {\bf g}$ and take a trace, then
$$
Tr(\ket {\bf f}\bra {\bf g}\check {\Phi }_t(a({\bf h})a^{\dag }({\bf e})))-\int \limits _0^tTr(\ket {\bf f}\bra {\bf g}{\mathcal M}(ds)\circ \check {\Phi }_{t-s}(a({\bf h})a^{\dag }({\bf e})))=
$$ 
\begin{equation}\label{anar5}
Tr(\ket {\bf f}\bra {\bf g}\check {\Phi } _t(a({\bf h})a^{\dag }({\bf e})))-\int \limits _0^tTr(\check {\Psi }_{t-s}\circ {\mathcal M}_*(ds)(\ket {\bf f}\bra {\bf g})a({\bf h})a^{\dag }({\bf e}))
\end{equation}
Suppose that $f_j,g_k\in D(d)$ and $h_j,e_k\in D(d_*)$. Taking the derivative at zero from (\ref {anar5}) we obtain
$$
Tr(\ket{\bf f}\bra {\bf g}(a(\hat d_*{\bf h})a^{\dag }({\bf e})+a({\bf h})a^{\dag }(\hat d_*{\bf e})))-Tr(\Delta (\ket{\bf f}\bra {\bf g})a({\bf h})a^{\dag }({\bf e}))
$$
due to (\ref {mera}) and (\ref {gen}). Now the result follows from Theorem 1.

\end{proof}

Taking into account (\ref {r1}), (\ref {r2}) and (\ref {r3}) we can conclude that for getting $\check {\Phi }_t$ to be unital the measure $\mathcal M$ creates a particle with the wave function $\chi _{r,r+dr}$ at each time moment $r$ preceding $t$.

\section{Conclusion}

We consider perturbations of a dynamical semigroup $\Phi $ on the algebra of all bounded operators determined by solutions of integral equations with respect to measures $\mathcal M$ on the semi-axis ${\mathbb R}_+$ with values in the set of completely positive maps which is covariant with respect to $\Phi $ such that $\Phi _r\circ {\mathcal M}([t,s))={\mathcal M}([t+r,s+r)),\ s,t,r\ge 0$. As an example we construct the perturbation of the semigroup of non-unital *-endomorphisms on the algebra of canonical anticommutation relations resulting in the flow of shifts.

\section*{Acknowledgments} This work is supported by the Russian Science Foundation under grant 19-11-00086.

\end{document}